\documentclass[12pt,draftcls,onecolumn]{IEEEtran}

\usepackage{amsfonts}
\usepackage{times}
\usepackage{latexsym}
\usepackage{amssymb}
\usepackage{amsmath}
\usepackage{cite}
\usepackage{verbatim}
\def\bb0{{\mathbb{0}}}

\def\bb{{\mathbf{b}}}

\def\b0{{\mathbf{0}}}

\def\b1{{\mathbf{1}}}

\def\bbE{{\mathbb{E}}}

\def\cA{\mathcal{A}}

\def\sf0{{\mathsf{0}}}

\def\nn{\nonumber}

\usepackage{mathtools,dsfont}
\usepackage{psfrag,bbm,graphicx}
\usepackage{comment,balance}
\usepackage{epstopdf}
\usepackage{epsfig}
\usepackage{multicol}
\usepackage{multirow, cite}
\usepackage{amsfonts}
\usepackage{color}
\usepackage{amsmath}
\usepackage{amsfonts,amsmath,amssymb,mathrsfs}
\usepackage{graphics, graphicx}

\usepackage{tikz}
\usetikzlibrary{shadows,patterns,shapes,backgrounds,automata}
\usetikzlibrary{fit,decorations.pathmorphing,plotmarks} % noisy shapes
\tikzstyle{every picture}+=[remember picture]

\usepackage[ruled,noline]{algorithm2e}
\usepackage[utf8]{inputenc}
\usepackage[english]{babel}
\newtheorem{lemma}{Lemma}
\newtheorem{theorem}{Theorem}
\newtheorem{definition}{Definition}
\newtheorem{remark}{Remark}
\newtheorem{assumption}{Assumption}

\begin{document}
\title{Online energy efficient packet scheduling for a common deadline with and without energy harvesting}

\author{\IEEEauthorblockN{Aditya Deshmukh}

\IEEEauthorblockA{Department of Electrical Engineering\\
Indian Institute of Technology Madras\\
Chennai 600036\\
Email: ee12b070@ee.iitm.ac.in}

\and

\IEEEauthorblockN{Rahul Vaze}

\IEEEauthorblockA{School of Technology and Computer Science\\
Tata Institute of Fundamental Research\\
Homi Bhabha Road, Mumbai 400005\\
Email: vaze@tcs.tifr.res.in}}

\maketitle

\begin{abstract}
The problem of online packet scheduling to minimize the required conventional grid energy for transmitting a fixed number of packets given a common deadline is considered. The total number of packets arriving within the deadline is known, but the packet arrival times are unknown, and can be arbitrary. The proposed algorithm tries to finish the transmission of each packet assuming all future packets are going to arrive at equal time intervals within the left-over time. The proposed online algorithm is shown to have competitive ratio that is logarithmic in the number of packet arrivals. The hybrid energy paradigm is also considered, where in addition to grid energy, energy is also available via extraction from renewable sources. The objective here is to minimize the grid energy use. A suitably modified version of the previous algorithm is also shown to have competitive ratio that is logarithmic in the number of packet arrivals.\end{abstract}

\IEEEpeerreviewmaketitle

\section{Introduction}

Minimizing energy consumption under delay constraints is a classical resource allocation problem, which has been very well studied in literature \cite{uysal2002energy, miao2006optimal, zafer2005calculus, chen2007energy, agarwal2008structural, berry2013optimal, lee2013asymptotically, neely2011opportunistic, neely2007optimal, srivastava2010energy}. Shorter the delay, larger is the energy required, and clearly, there is a tradeoff between these two fundamental quantities. The problem is even more important for the modern setup, where devices are increasing becoming smaller with limited battery sizes, and where effective energy utilization is fundamental for efficient network operation by maximizing the node lifetime and consequently expanding the network lifetime.

The energy-delay tradeoff has been studied in variety of settings. For example, for an AWGN channel, \cite{uysal2002energy} studies the packet scheduling problem for minimizing transmission energy, where a fixed number of packets arrive successively in time, and have to be transmitted before a common deadline. 
Assuming that packet arrival times are known ahead of time (called the offline setting), an optimal algorithm is derived in \cite{uysal2002energy}.
A more general problem than studied in \cite{uysal2002energy},  is where each packet has an individual hard deadline \cite{miao2006optimal, zafer2005calculus}. With individual hard deadlines, the optimal algorithms are known in the offline setting \cite{miao2006optimal, zafer2005calculus}, or when the packet arrival times are i.i.d. and follow a given distribution \cite{chen2007energy}.  

Similar results are available for fading/time-varying channels, where typically an average delay constraint is considered \cite{agarwal2008structural, gong2013optimal, srivastava2010energy}, and the problem is to minimize the average power consumption. A hard deadline result is also known from \cite{lee2013asymptotically}. A more specific case of energy-delay tradeoff with 'small' delay constraint has been addressed in \cite{berry2013optimal} for the fading channels.

The energy-delay tradeoff problem becoming even more challenging, when in addition to the conventional grid energy there is an additional energy source that is powered by renewable energy, popularly called as energy harvesting (EH). The problem of minimizing transmission time/delay when only an EH source is available has been well studied in literature. For example, for the AWGN channel, an optimal offline algorithm has been derived for a single transmitter-receiver pair in \cite{Yang2012}, whose online counterpart that is $2$-competitive for arbitrary energy arrivals has been found in \cite{VazeCompetitive}. With only an EH source, the problem of multiple packet transmissions with individual hard deadlines has been studied recently in \cite{shan2015discrete}. Similarly, for the fading channel, an optimal offline algorithm has been derived for a single transmitter-receiver pair in \cite{ozel2011transmission}, while for broadcast and MAC channels in \cite{antepli2011optimal, yang2012optimal, ozel2012optimal, Antepli}.

The problem where both the grid and the EH energy are concurrently available is relatively less well-studied and is inherently a hard problem. Starting with \cite{cui2012delay}, some progress has been made in \cite{gong2013optimal}, where optimal offline (under some conditions on battery size etc.), and two heuristic online algorithms are derived. The general online problem has remained unsolved. 

In this paper, when only grid energy is available, we first consider the classical packet scheduling problem for minimizing transmission energy \cite{uysal2002energy}, where a fixed number of packets $P$ arrive successively in time, and have to be transmitted before a common deadline $T$. In contrast to the offline case studied in \cite{uysal2002energy}, in this paper, we consider the more realistic and challenging {\it online} scenario, where information about packet arrivals is known only causally. Moreover, unlike \cite{chen2007energy}, we do not make any stochastic assumptions on the inter-arrival times for the packets, and consider the arbitrary packet inter-arrival time case, where even an adversary can choose them. Thus, our model is the most general one, and is applicable for any setting where packet inter-arrival times are time-variant or difficult to estimate etc.

To characterize the performance of an online algorithm, we consider the metric of competitive ratio that is defined as the ratio of the energy used by the online algorithm and the offline optimal algorithm, maximized over all packet inter-arrival time sequences. The competitive ratio is a worst case guarantee on the performance of an online algorithm and is independent of modeling assumptions.  

We assume that at $t=0$, the number of packets $P$ and common deadline $T$ are known. We show in Remark \ref{rem:unboundedcr}, that if $P$ is not known ahead of time, then the competitive ratio of any online algorithm is unbounded.
Let on arrival of a new packet at time $t$, the number of packets left to arrive be $P(t)$. Then the main idea behind the proposed algorithm $\mathsf{ON}$ is that it assumes that the future $P(t)$ packets are going to arrive at equal intervals in the left-over time of $T-t$, and attempts to finish transmitting the current packet in time $\frac{T-t}{P(t)+1}$. Since the future inter-arrival time sequence is unknown and arbitrary, algorithm $\mathsf{ON}$ may have to idle, i.e. it can finish transmitting the current packet before the next packet arrives, in which case it has to use more energy than required by the optimal offline  algorithm, that never idles. We show that the competitive ratio of the $\mathsf{ON}$ algorithm is no more than $1+ \log P$, where $P$ is the total number of packets. 

Note that the derived competitive ratio is independent of the common deadline time $T$, and the number of bits contained in each packet. Ideally, we would also have liked the competitive ratio to be independent of $P$ and be a constant, however, for many similar scheduling and load-balancing {\it online} problems, the best results on competitive ratio also scale logarithmically in the number of packets/users (equivalent quantity of interest) \cite{azar1994line, aspnes1997line, gobel2014online}. We would like to note that $1+ \log P$ is the best bound we can show theoretically for the $\mathsf{ON}$ algorithm, however, 
 simulations suggest that the competitive ratio of $\mathsf{ON}$ is far better than that and is close to $1$ for the examples considered. 
 
Our derived results are for the most general input setting for this classical problem, and to the best of our knowledge no online algorithms with provable guarantees on the competitive ratios are known in the literature.

Next, we generalize the energy arrival paradigm, and consider the same problem of minimizing energy for transmitting multiple packets given a common deadline, when energy from both the grid and the EH source is available. In this hybrid energy source scenario, if the energy harvesting profile is arbitrary, then it is easy to see that no online algorithm can have bounded competitive ratio, since if large amount of EH energy arrives close to the deadline, the offline algorithm will use it intelligently, while the online algorithm may not. Thus, for this case, one has to make the assumption that the EH energy profile is stochastic, and energy arrivals are identically distributed across time. The inter-arrival times are still allowed to be arbitrary.
We propose a natural greedy extension of the $\mathsf{ON}$ algorithm, that uses the EH energy as quickly as possible and for as long as possible while keeping the same transmission times for each packet as prescribed by the $\mathsf{ON}$ algorithm. Similar to the only grid energy case, we show that the competitive ratio of our algorithm in this hybrid energy scenario is bounded by $c(1+\log P)$ for constant $c >1$. Using numerical results, we conclude that the competitive ratio of the proposed algorithm  is actually very close to $1$ for the considered examples, and it is expected to do well in the online setting.

\section{Model}
We consider a single transmitter-receiver pair, that wants to communicate $P$ packets that arrive within time 
$[0,T)$, with a common deadline of $T$ for all $P$ packets, i.e., all packets should be delivered by time $T$. The number of bits in each packet is assumed to be equal to $B$.
The transmitter is connected to two sources of energy through which it extracts power: i) the grid (conventional), and  ii) a battery that is replenished by a energy harvester that is powered by a renewable energy source. 
Naturally, there is a cost associated to the grid energy usage, whereas renewable energy is available at zero cost. Thus, the objective is to minimize the total grid energy used to transmit the $P$ packets by common deadline time $T$.

We use the Shannon formula $B = t \log \left(1+\frac{E}{t}\right)$ to find the energy needed to send $B$ bits in time duration $t$, as 
\begin{equation}\label{eq:energyusage} 
f(t) = t(2^{B/t}-1). \footnote{More generally with noise power $N_0$ and bandwidth $W$, $f(t) = N_0Wt(2^{B/(tW)}-1))$.}
\end{equation}
The rate of power transfer is denoted as $R = \frac{E}{t}$.

We assume that the first packet arrives at $t=0$, and the inter-arrival time between the $i^{th}$ and the 
$(i+1)^{th}$ packet is given by $a_{i}$. Thus, a packet arrival sequence is represented as a sequence : 
$$A_{P}=(a_{1},a_{2},a_{3},...,a_{P-1}, a_P),$$ where, $a_{i}\geq0$ and $\sum\limits^{P-1}_{i=1}a_{i} < T$ and $\sum\limits^{P}_{i=1}a_{i} = T$. We have introduced the extra time $a_P$ that accounts for the time difference between the last ($P^{th}$) packet arrival at time $\sum_{i=1}^{P-1}a_{i}$ and $T$. See Fig. \ref{fig:on-optdescription} for an illustration.
Let $\Delta_{P}^{T}$ be the set of sequences representing packet inter-arrival times with number of packets equal to $P$, i.e., 
$$\Delta_{P}^{T}=\left\{ \left(a_{1},a_{2},a_{3},...,a_{P}\right)\ | \ a_{i}\geq0,\sum\limits _{i=1}^{P}a_{i} =T\right\}.$$
Since $P$ and $T$ are fixed, we will use just $A$ and $\Delta$ instead of $A_{P}$ and $\Delta_{P}^{T}$ for simplicity.

Following prior work and to keep the system complexity low, we assume that bits from different packets cannot be transmitted at the same. Thus, packets are transmitted one after another in a sequential fashion. 
\begin{definition} For packet $i$, let $s_i\ge \sum_{j=1}^{i-1}a_{j}$ and $f_i$ be the start and the finish time of transmission of packet $i$, respectively. Then we define $t_i=f_i-s_i$ to be the transmission time for packet $i$. \end{definition}
\begin{definition}
With packet transmission times $t_i$, the total energy used by an online algorithm $\text{ALG}$ to transmit $P$ packets is given by\footnote{For ease of exposition we have indexed packets as $\{0,\ldots, P-1\}$}
$$E_{\text{ALG}}\left(A\right)= \sum_{i=0}^{P-1} f(t_i).$$ To distinguish the optimal offline algorithm $\mathsf{OPT}$ from any online algorithm, let $d_i$ be the packet transmission times of $\mathsf{OPT}$, and total energy used by $\mathsf{OPT}$ to transmit the $P$ packets be
$$E_{\mathsf{OPT}}\left(A\right)= \sum_{i=0}^{P-1} f(d_i).$$
\end{definition}

\begin{definition}
The competitive ratio of algorithm $\text{ALG}$ is defined as :
$$\mu_{\text{ALG}}=\max_{A \in \Delta}\frac{E_{\text{ALG}}(A)}{E_{\mathsf{OPT}}(A)},$$
where $\mathsf{OPT}$ is the  optimal offline algorithm.
\end{definition}
The competitive ratio is the worst case ratio of the cost of the online algorithm and the optimal offline algorithm over all possible inter-arrival sequences, and has been used extensively to quantify the performance of online algorithms.

We first consider the scenario where no energy harvesting is available, and the objective is to minimize the grid energy usage. In prior work \cite{uysal2002energy}, this problem has been addressed in the offline scenario, where the inter-arrival time sequence $A$ is revealed ahead of time, non-causally. We consider a more realistic online scenario, where $A$ is revealed causally, and where $A$ can be arbitrary with no distribution information. 
To keep the problem non-degenerate, however, we assume  that the number of packets $P$ is known ahead of time. In Remark \ref{rem:unboundedcr}, we show that if $P$ is unknown, even if $P$ can take only two values $\in \{1,2\}$, the competitive ratio is unbounded for any online algorithm. 

\begin{remark}\label{rem:unboundedcr} If the number of packets $P$ is not known ahead of time, then we show that the competitive ratio of any online algorithm is unbounded. Consider the case when $P$ is either $1$ or $2$, but that is unknown to the online algorithm ahead of time. Also let $T=1$. If $P=2$, then let $a_1= \frac{1}{2}$. Then until time $t=\frac{1}{2}$, any online algorithm does not know whether $P=1$ or $P=2$. Without this information, let any online algorithm spend energy $E_1$ until time $t= \frac{1}{2}$ and transmit $B_0\le B$ bits of the first packet. 

Using the energy function $f(t) = t(2^{B/t}-1)$, if $P=1$, i.e., no packet arrives at time $t=\frac{1}{2}$, then the total energy used by the online algorithm is given by  $\frac{1}{2}(2^{2 B_0}-1) + \frac{1}{2}(2^{2 (\max\{B-B_0,0\})}-1)$, while if $P=2$, it is 
$\frac{1}{2}(2^{2 B_0}-1) + \frac{1}{2}(2^{2 (2B-B_0)}-1)$, since from time $t=\frac{1}{2}$ onwards it has to transmit $2B-B_0$ bits in the left-over time interval $[\frac{1}{2},1]$. 

Moreover, since the optimal offline algorithm knows the exact value of $P$ ahead of time, the total energy it spends is $(2^{B}-1)$ if $P=1$, and $(2^{2B}-1)$ if $P=2$, since it transmits first packet completely by time $\frac{1}{2}$ knowing that $P=2$. Thus, the competitive ratio of any online algorithm $\text{ALG}$ is lower bounded by 
\begin{eqnarray*}
\mu_{\text{ALG}} &\ge& \min_{B_0} \max\left\{ \frac{\frac{1}{2}(2^{2 B_0}-1) +  \frac{1}{2}(2^{2 (B-B_0)}-1)}{2^{B}-1},  \right.\\
&&\left.\frac{\frac{1}{2}(2^{2 B_0}-1) + \frac{1}{2}(2^{2 (2B-B_0)}-1)}{2^{2B}-1}\right\}.
\end{eqnarray*}
It is easy to see that for any value of $B_0$ that the online algorithm chooses, the competitive ratio grows exponentially in $B$.
\end{remark}

\section{No Energy Harvesting}
In this section, we consider the case when no renewable source is available, and the objective is to minimize the grid energy usage for transmitting the $P$ packets. Supposing that the inter-arrival time sequence $A$ is known ahead of time, an optimal offline algorithm has been derived in \cite{uysal2002energy}, which we present here for completeness sake, as well for an easier presentation of our online algorithm.

\subsection{Optimal Offline Algorithm}
\begin{algorithm}
	\caption{$\mathsf{OPT}$}
	initialize $k_{0}=0$\;
	\For{$j := 0$ to $P-1$}
	{
	$m_{j+1}=\underset{k\in\{1,2,3,...,P-k_{j}\}}{\max}\left\{\frac{1}{k}\sum\limits^{k}_{i=1}a_{k_{j}+i}\right\}$
 
	$k_{j+1}=\underset{k\in\{1,2,3,...,P-k_{j}\}}{\max}\left\{k:\frac{1}{k}\sum\limits^{k}_{i=1}a_{k_{j}+i}=m_{j+1}\right\}$
	}
	\For{$i := 0$ to $P-1$}
	{
		$d_{i}=m_{j}$ such that $k_{j-1}<i\leq k_{j}$\;
	}
	
	return $(d_{0},d_{1},...,d_{P-1})$\;
\end{algorithm}
The optimum offline algorithm $\mathsf{OPT}$ for minimizing the total energy for transmitting $P$ packets with a common deadline is given by Algorithm 1  \cite{uysal2002energy}. From \eqref{eq:energyusage}, it is clear that transmitting at a slower rate (power), minimizes the energy required. Since sequence $A$ is known ahead of time, 
the $\mathsf{OPT}$ algorithm makes sure that the transmitter never idles by transmitting at rate (slower/faster) depending on the next packet arrival times (large/short). 

The offline algorithm computes the largest average $m_{1}$ of partial sums of $a_{i}$'s starting from index 
$i=1$ to $P$, and sets the first transmission time, i.e. packet transmission finish time, equal to 
$m_{1}$ for each of the first $k_{1}$ number of packets, where $k_{1}$ is the highest index such that the average of partial sums of $a_{i}$'s is $m_{1}$. It then repeats the same procedure after index $k_{1}$. The algorithm never idles and $i^{th}$ packet is transmitted immediately after the transmission of the  $(i-1)^{th}$ packet ends.

 The packet transmission times output by $\mathsf{OPT}$, $d_{i}$'s, are such that $\sum\limits _{i=0}^{P-1}d_{i}=T$, since $A$ is known ahead of time, and the algorithm can ensure the non-idling property.
 
 Also, since the transmission of $i^{th}$ packet cannot start before its arrival,  we also have that, 
\begin{equation}
\label{ineq:causal}
\sum\limits _{i=1}^{\ell}d_{i}\geq\sum\limits _{i=1}^{\ell}a_{i},
\end{equation}
for any $0\le \ell \le P-1$.
Moreover, another useful property of  $\mathsf{OPT}$ \cite{uysal2002energy} is that, 
\begin{equation}
\label{ineq:tau}
d_{i}\geq d_{i+1}\;\forall i=0, 1, \dots P-1,
\end{equation}
i.e., the transmission times decrease with the index of the packets, which is intuitive, since otherwise we could stretch the transmission time and decrease the energy usage. We will make use of \eqref{ineq:tau} repeatedly while analyzing the competitive ratio of the $\mathsf{ON}$ algorithm.

%\begin{algorithm}
%\caption{CH election algorithm}
%\label{CHalgorithm}
%\begin{algorithmic}[1]
%\Procedure{\mathsf{OPT}} {}
%\For{$j := 1$ to $P-1$ }
%\state $m_{j+1}=\underset{k\in\{1,2,3,...,P-k_{j}\}}$
%\state $k_{j+1}=\underset{k\in\{1,2,3,...,P-k_{j}\}}{\max}\left\{k:\frac{1}{k}\sum\limits^{k}_{i=1}a_{k_{j}+i}=m_{j+1}\right\}$
%\\QOS($k$) = max {QOS($j$) \textbar $j$ \Pisymbol{psy}{206} $N1$($i$)  U $i$}
%\\ MPRSet($i$) = $k$
%\EndFor
%\EndProcedure
%\end{algorithmic}
%\end{algorithm}

Next, we present an important property of the  $\mathsf{OPT}$ algorithm that will be useful for the analysis of our online algorithm.
\begin{lemma}\label{lem:OPTdectime} \cite{uysal2002energy} If the inter-arrival time sequence $A$ is such that $a_i \ge a_{i+1}, \ \forall \ i=1,\dots, P-1$, 
then, $d_{i-1}( \mathsf{OPT}) = a_i$, i.e., the optimal offline algorithm finishes each packet exactly at the arrival time of the next packet. 
\end{lemma}
\begin{proof} Note that the  $\mathsf{OPT}$ algorithm computes the largest averages of partial sums of $a_i$'s in each round. For the case when $a_i \ge a_{i+1}$, then in each iteration of $\mathsf{OPT}$,  trivially, $d_{i-1}=a_i$, by the definition of the $\mathsf{OPT}$.
\end{proof}

Now, we describe our online algorithm called $\mathsf{ON}$ and then derive its competitive ratio. 

\subsection{Online Algorithm $\mathsf{ON}$}\label{sec:ON}
In light of Remark \ref{rem:unboundedcr}, we assume that at $t=0$, the number of packets $P$ and common deadline $T$ are known. 
Let on arrival of a new packet at time $t$, the number of packets left to arrive be $P(t)$. Then the main idea behind the algorithm is that it assumes that the future $P(t)$ packets are going to arrive at equal intervals in the left-over time of $T-t$, and attempts to finish transmitting the current packet in time $\frac{T-t}{P(t)+1}$ (transmission time). Since the inter-arrival time sequence $A$ is unknown and arbitrary, this algorithm may have to idle, i.e. it can finish transmitting the current packet before the next packet arrives, in which case it has to use more energy than required by the $\mathsf{OPT}$ algorithm, that never idles. We later show that the competitive ratio of $\mathsf{ON}$ is no more than $1+\log P$. 

A more formal description of the algorithm is as follows.
The transmitter starts sending the first packet at time $0$ with transmission time $\frac{T}{P}$. If the second packet arrives before the finish time of the first packet, the second packet is added to the queue and waits for the current transmission to complete. Once the first packet transmission is complete at time $\frac{T}{P}$, the second packet is transmitted with transmission time of 
$\frac{T}{P}$. Similarly, for the $i^{th}$ packet: if packet arrives before time $t=(i-1)\frac{T}{P}$, its added to the queue and transmitted starting from the time at which the $(i-1)^{th}$ packet's transmission got completed with transmission time $\frac{T}{P}$. If suppose, the $j^{th}$ packet arrives after the finish time of the $(j-1)^{th}$ packet. Then for the time between the arrival of the $j^{th}$ packet and the finish time of the $(j-1)^{th}$ packet, the transmitter has no packets in the queue and is said to be 'idle', and does not consume any power. 
In such a case, at the time of the arrival of the $j^{th}$ packet at time $\sum_{i=1}^{j}a_{i}$, we update :
$$T \leftarrow T - \sum_{i=1}^{j}a_{i},$$
$$P \leftarrow P - j.$$
The algorithm now repeats the same procedure with the new $T$ and $P$, and outputs packet transmission times $t_i, i=j,\dots, P-1$.

With algorithm $\mathsf{ON}$, the transmission time for the $i^{th}$ packet is given by : $t_{i}=\min\left(t_{i-1},\, \frac{T-\sum\limits _{n=1}^{i}a_{n}}{P-i}\right)$ or equivalently,
$$t_{i}=\min\left(\frac{T}{P},\frac{T-a_{1}}{P-1},\frac{T-a_{1}-a_{2}}{P-2},...,\frac{T-\sum\limits _{n=1}^{i}a_{n}}{P-i}\right).$$

Let the ratio of the remaining time and the number of packets yet to arrive at the $\ell^{th}$ packet arrival be 
$x_{\ell}=\frac{T-\sum\limits _{n=1}^{\ell}a_{n}}{P-\ell}\;\;\forall\,0\le \ell\leq P-1$. Then $t_{i}$ can be expressed as : 
\begin{equation}
\label{eq:x}
t_{i}=\underset{\ell\leq i}{\min}\left(x_{\ell}\right).
\end{equation}

\begin{algorithm}
	\caption{$\mathsf{ON}$}
	initialize $t_{0}=\frac{T}{P}$\;
	\For(){$i := 0$ to $P-2$}
   	{
	$t_{i+1}=\min\left(t_{i},\frac{T-\sum\limits^{i}_{n=1}a_{n}}{P-i}\right)$\; 
    }
    %$E_{\text{\text{on}}}(A)=\sum\limits^{P-1}_{i=0}f(t_{i})$ \;
    return $(t_0, t_{1},t_{2},t_{3},...,t_{P-1})$\;
\end{algorithm}

Compared to the $\mathsf{OPT}$ algorithm, $\mathsf{ON}$ will pay a penalty, if the inter-arrival times are much larger than its assumption of them being equally spaced. The penalty arises because in such cases, $\mathsf{ON}$ has to idle, and consequently transmit packets in shorter time consuming larger energy compared to the $\mathsf{OPT}$ algorithm. We make this intuition concrete in Theorem \ref{th1}, where we show that the worst case input that maximizes the competitive ratio for $\mathsf{ON}$ is of the type when the inter-arrival times are decreasing, i.e., $a_i \ge a_{i+1}, i=0,\dots,P-1$. When $a_i \ge a_{i+1}$, $\mathsf{ON}$ has to idle after finishing every packet transmission. To  see this, with $a_i \ge a_{i+1}$, $a_1>T/P$, hence $\mathsf{ON}$ idles from time $T/P$ (where it finishes the first packet transmission) till $a_1$. From time $a_1$ onwards,   $\mathsf{ON}$ treats $a_1$ as time $0$ and restarts the process, and hence has to idle after finishing each packet transmission. 

In Fig. \ref{fig:on-optdescription}, we give a concrete example of transmission times set by the $\mathsf{ON}$ algorithm and the offline optimal algorithm $\mathsf{OPT}$ for a particular sequence $A$. The gaps in time-line for $\mathsf{ON}$ are because of its possible idling which $\mathsf{OPT}$ completely avoids. 
\begin{figure}[h]
	\begin{center}
		\includegraphics[height=2.75in]{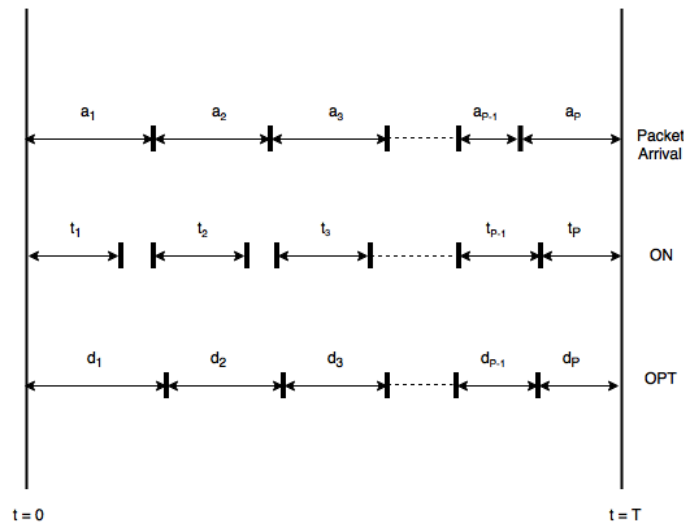}
		\caption{\sl Illustration of $\mathsf{ON}$ and $\mathsf{OPT}$ for a particular inter-arrival sequence $A$. \label{fig:on-optdescription}} 
	\end{center}
\end{figure}

\subsection{Competitive Ratio Analysis of $\mathsf{ON}$}
We now show that the worst case inter-arrival sequence ($a_{i}$'s) for algorithm $\mathsf{ON}$ is such that $a_{i}\geq a_{i+1}$. This condition essentially implies that the $\mathsf{ON}$ algorithm has to idle after finishing transmission of every single packet. To prove this, we show that given any inter-arrival sequence $A$, we can construct another feasible inter-arrival sequence $A'$ for which $a_{i}\geq a_{i+1}$ and $E_{\mathsf{OPT}}(A) = E_{\mathsf{OPT}}(A')$, while the energy spent by $\mathsf{ON}$ increases in the latter case, i.e, $E_{\mathsf{ON}}(A) \le E_{\mathsf{ON}}(A')$.

%Let $\mathsf{ON}(A)$ and $\mathsf{OPT}(A)$ denote the output, that is the packet transmission times, set by $\mathsf{ON}$ and $\mathsf{OPT}$ respectively, for the packet inter-arrival time sequence $A \in \Delta$.

\begin{lemma}
\label{lem1}
Let $A^{\prime}$ be the inter-arrival sequence that is output (packet transmission times) of the optimal offline algorithm $\mathsf{OPT}$ with inter-arrival sequence $A$, i.e., $A^{\prime}\leftarrow {\mathsf{OPT}}(A)$.
Then, we have
\begin{eqnarray}\label{eq:ONincdec}
E_{\mathsf{ON}}(A^{\prime})&\geq &E_{\mathsf{ON}}(A),\\\label{eq:OPTnochange}
E_{\mathsf{OPT}}(A^{\prime})&=&E_{\mathsf{OPT}}(A).
\end{eqnarray}

\end{lemma}
\begin{remark}
Note that $A$ is a $P$-length sequence, but the last element of $A$ is auxiliary, since packet $0$ starts at time $0$, and only the first $P-1$ elements represent the inter-arrival times of $P-1$ other packets. The output $\{d_0, \dots, d_{P-1}\}$ of ${\mathsf{OPT}}(A)$ is of length $P$, with transmission times for the successive $P$ packets. So when we consider output of ${\mathsf{OPT}}(A)$ as an input to $\mathsf{ON}$  or $\mathsf{OPT}$, we mean that $a_{i+1}= d_{i}$ for $i=0,\dots, P-1$ and $a_P = T- \sum_{i=1}^{P-1} a_{i}$.
\end{remark}
\begin{remark}
From \eqref{ineq:tau}, we have that elements of $A'$ are such that $d_0\ge d_1 \ge \dots \ge d_{P-1}$. Thus, Lemma \ref{lem1} shows that with decreasing inter-arrival sequences, the competitive ratio increases for $\mathsf{ON}$.
\end{remark}
Now we present the proof of Lemma \ref{lem1}.
\begin{proof} 
We prove this by showing that the packet transmission times for $\mathsf{OPT}$ remain the same with $A$ or $A'$, whereas they decrease for $\mathsf{ON}$ with $A'$ in comparison to $A$. 
Let, $A^{\prime}=(d_{0},d_{1}, \dots,d_{P-1})$. Note that $d_{i}\geq0$ and $\sum\limits^{P-1}_{i=0}d_{i}=T$, hence as explained in Remark \ref{lem1}, $A^{\prime}$ is a valid packet arrival sequence, and therefore $A^{\prime}\in\Delta$.

Let, $\left(t_{i}\right)_{i=1}^{P}$ and $\left(t_{i}^{\prime}\right)_{i=1}^{P}$ be the packet transmission times set by $\mathsf{ON}$ for $A$ and $A^{\prime}$, respectively.

We first prove that the energy consumed by $\mathsf{ON}$ acting on $A^{\prime}$ is greater with respect to $A$ by showing that the transmission times for each packet decrease in $A^{\prime}$, and hence the total energy increases with respect to $A$.
From \eqref{eq:x}, we know that for packet inter-arrival time sequence $A$, with $\mathsf{ON}$,
\begin{eqnarray*}
x_{\ell}&=&\frac{T-\sum\limits _{n=1}^{\ell}a_{n}}{P-\ell},\\
t_{i}&=&\underset{l\leq i}{\min}(x_{\ell}),
\end{eqnarray*}
and,
\begin{eqnarray*}
x_{\ell}^{\prime}&=&\frac{T-\sum\limits _{n=1}^{\ell}d_{n}}{P-\ell},\\
t_{i}^{\prime}&=&\underset{\ell\le i}{\min}(x_{\ell}^{\prime}).\\
\end{eqnarray*}
Therefore, using \eqref{ineq:causal}, we have
\begin{eqnarray*}
\frac{T-\sum\limits_{n=1}^{\ell}a_{n}}{P-\ell}&\geq&\frac{T-\sum\limits_{n=1}^{\ell}d_{n}}{P-\ell}, \;\,\forall \ell.\\
\end{eqnarray*}
Hence, by definition of $x_{\ell}$, $x_{\ell} \ge x_{\ell}^{\prime}, \ \forall \ell
$.
Therefore, 
\begin{eqnarray*}
\underset{l\leq i}{\min}(x_{l})&\geq &\underset{l\leq i}{\min}(x_{l}^{\prime})\;\;\forall i,\\
t_{i}&\geq &t_{i}^{\prime}\;\;\forall i,
\end{eqnarray*}
Since the energy function \eqref{eq:energyusage} is inversely proportional to transmission time, 
\begin{eqnarray*}
f(t_{i}^{\prime})&\geq &f(t_{i})\;\;\forall i,\\
\sum\limits^{P-1}_{i=0}f(t_{i}^{\prime})&\geq&\sum\limits^{P-1}_{i=0}f(t_{i}),\\
 E_{\text{\text{on}}}(A^{\prime})&\geq &E_{\text{\text{on}}}(A).
\end{eqnarray*}
Thus, we have proved \eqref{eq:ONincdec}.

From \eqref{ineq:tau}, we know that $A' = \{d_0, \dots, d_{P-1}\}$ is such that $d_i\ge d_{i+1}$. Therefore, from Lemma \ref{lem:OPTdectime}, we have that with $\mathsf{OPT}$, the  transmission times remain same for both packet sequences $A$ and $A^{\prime}$. Hence $\mathsf{OPT}$ uses identical energy for $A$ or $A^{\prime}$, proving \eqref{eq:OPTnochange}.

\end{proof}

Now, we show that the worst case packet sequence that maximizes the competitive ratio of $\mathsf{ON}$ is such that $a_{i}$'s are decreasing.
\begin{theorem}
\label{th1}
Let $\cA = \left\{A \in \Delta : \mu_{\mathsf{ON}}(A) \ge \mu_{\mathsf{ON}}(A'), \right.$ $\left. \ \forall \ A' \in \Delta\right\}$ be the set of inter-arrival time sequences that have the worst competitive ratio.
Let $\Delta_{D}\subseteq \Delta$  be such that
$$\Delta_{D}=\left\{ (a_{1},a_{2},\dots,a_{P}) \ | \ a_{i}\geq a_{i+1}, \sum_{i=1}^P a_i=T\right\},$$
$\forall i\in\{1,...,P\}$.
Then 
$$\cA \cap \Delta_D \ne \phi.$$
\end{theorem}
Theorem \ref{th1} implies that at least one of the worst inter-arrival sequences belongs to set $\Delta_{D}$.

\begin{proof} Let $A \in \cA$. Then consider $$A^{new}\leftarrow \mathsf{OPT}\left(A \right),$$ i.e., $A^{new}$ is  the output of the 
$\mathsf{OPT}$ given the input $A \in \cA$ for $\mathsf{ON}$. 
Note that $A^{new}\in \Delta_{D}$ from \eqref{ineq:tau}.
Using Lemma \ref{lem1}, we have, 
$E_{\mathsf{ON}}\left(A^{new}\right)\geq E_{\mathsf{ON}}\left(A \right)$, while $E_{\mathsf{OPT}}\left(A^{new}\right)=E_{\text{\text{off}}}\left(A \right)$ from Lemma \ref{lem:OPTdectime}.

Hence,
$$\frac{E_{\mathsf{ON}}(A^{new})}{E_{\mathsf{OPT}}(A^{new})}\geq\frac{E_{\mathsf{ON}}(A )}{E_{\mathsf{OPT}}(A )},$$ and in particular  
$$\frac{E_{\mathsf{ON}}(A^{new})}{E_{\mathsf{OPT}}(A^{new})}\geq\frac{E_{\mathsf{ON}}(A' )}{E_{\mathsf{OPT}}(A' )},$$ for any $A' \in \Delta$ by the definition of $\cA$. 
 Therefore $A^{new}$ also belongs to $\cA$, and $$\cA \cap \Delta_D \ne \phi.$$
%This is a contradiction to the assumption that $A^{*}\notin \Delta_{D}$, since $A^{*}$ is the worst case input for $\mathsf{OPT}$. Hence, we have shown that the worst case for our online algorithm is when $a_{i}$ are decreasing. 
%This is justified in the sense that when using the online algorithm in the case where $a_{i}$ are decreasing, the transmitter is idle after every packet transmission, whereas the offline is never idle.
\end{proof}

We now prove a useful result about packet transmission times set by the $\mathsf{ON}$ algorithm when $a_{i}$'s are decreasing.

\begin{lemma}
\label{lem2}
If the inter-arrival time sequence $A=\left(a_{i}\right)_{i=1}^{P}\,\in \Delta_{D}$, then 
$\{t_{i}\}_{i=0}^{P-1}=\mathsf{ON}(A)$ is such that 
\begin{eqnarray*}
t_{i}=\frac{T-\sum\limits^{i}_{n=1}a_{n}}{P-i}. 
\end{eqnarray*}
\end{lemma}

\begin{proof}
Recall that for $A=\left(a_{i}\right)_{i=1}^{P}\,\in \Delta_{D}$, we have,
$a_{1}\geq a_{2}\geq \dots \geq a_{P}$.
Therefore, from \eqref{eq:x},
$$t_{i}=\underset{l\leq i}{\min}(x_{\ell}),$$ where $x_{\ell}=\frac{T-\sum\limits _{n=1}^{\ell}a_{n}}{P-\ell}$.
Let $T^{\prime}=T-\sum\limits _{n=1}^{\ell}a_{n}$ and $P^{\prime}=P-\ell$.
Hence, $x_{\ell}=\frac{T^{\prime}}{P^{\prime}}$, and consider $x_{\ell+1}=\frac{T-\sum\limits _{n=1}^{\ell+1}a_{n}}{P-\ell-1}=\frac{T^{\prime}-a_{\ell+1}}{P^{\prime}-1}$.

Note that $ a_{\ell+1}\geq a_{\ell+2}\geq...\geq a_{P-1} \ge a_{P}$.
Hence, $\left(P-\ell\right)a_{\ell}\geq\sum\limits _{n=\ell+1}^{P}a_{n}=T-\sum\limits _{n=1}^{\ell}a_{n}, \ P^{\prime}a_{\ell+1}\geq T^{\prime}$ or equivalently, $a_{\ell+1}\geq\frac{T^{\prime}}{P^{\prime}}$.
As a result, $$x_{\ell+1}=\frac{T^{\prime}-a_{\ell+1}}{P^{\prime}-1}\leq\frac{T^{\prime}-\frac{T^{\prime}}{P^{\prime}}}{P^{\prime}-1}=\frac{T^{\prime}}{P^{\prime}}=x_{\ell}.$$
Hence, $t_{i}=\underset{\ell\leq i}{\min}\:\left(x_{\ell}\right)=x_{i}=\frac{T-\sum\limits _{n=1}^{i}a_{n}}{P-i}$.

\end{proof}

\subsection{Competitive Ratio Computation}
We are now ready to compute an upper bound on the competitive ratio of the $\mathsf{ON}$ algorithm, by making use of Theorem \ref{th1}, that states that the worst case arrival sequence for $\mathsf{ON}$ is when $a_i \ge a_{i+1}$.

\begin{theorem}
\label{thm:main}
The competitive ratio of the $\mathsf{ON}$ algorithm is upper bounded by 
$$\mu_{\mathsf{ON}} \le 1+\log\left(P\right).$$
\end{theorem}

\begin{proof}
We prove the Theorem via induction on the number of packets $P$.
%Consider the case of $P=2$, i.e., where the first packet is available at time $0$, and the second packet arrives at time $a_1< T$.  
For ease of exposition, we index the inter-arrival sequence $A$ by the the number of packets it contains, i.e., $A = A_k$ if the number of packets in $A$ are $k$. Using Theorem \ref{th1}, we will only consider inter-arrival sequences belonging to $\Delta_D$.
Consider $P=1$, where the first packet is available at time $0$, and no more packets arrive thereafter. Hence both the $\mathsf{ON}$ and the $\mathsf{OPT}$ algorithms use the same energy to transmit one packet, and 
$\mu_{\text{\text{on}}}\left(A_{1}\right)=1$.

Now, assume that the result holds for any sequence $A_k$ of $k$ packets, i.e.,
$\mu_{\mathsf{ON}}\left(A_{k}\right)\leq1+\log\left(k\right)$,
and consider a sequence $A_{k+1}$ of $k+1$ packets.

Let $\cA(k+1) = \left\{A \in \Delta(k+1) : \mu_{\mathsf{ON}}(A) \ge \mu_{\mathsf{ON}}(A')\right.$ $\left. \forall \ A' \in \Delta(k+1)\right\}$, where 
$$\Delta(k+1) = \left\{ \left(a_{1},a_{2},\dots,a_{k}, a_{k+1}\right)|a_{i} \ge 0,\sum\limits _{i=1}^{k+1}a_{i} =T\right\},$$ i.e., all possible inter-arrival sequences with $k+1$ packets. Recall that $a_{k+1}$ is auxiliary since packet $0$ arrives at time $0$.

Let $A_{k+1}^{*} \in \cA(k+1) \cap \Delta_{D}({k+1})$, where 
$\Delta_D(k+1) \subseteq \Delta(k+1)$ with $a_{i}\ge a_{i+1}, \ \forall \ i$. The set $\cA(k+1) \cap \Delta_{D}({k+1})\neq\phi$ from Theorem \ref{th1}. 

Consider the output of $\mathsf{ON}$ and $\mathsf{OPT}$, if the input inter-arrival sequence is $A_{k+1}^{*}$, i.e.,
let
\begin{align}\label{eq:k+1ontime}
\left(t_{i}(k+1)\right)_{i=0}^{k}&\leftarrow \mathsf{ON}\left(A_{k+1}^{*}\right)\\\label{eq:k+1OPTtime}
\left(d_{i}(k+1)\right)_{i=0}^{k}&\leftarrow \mathsf{OPT}\left(A_{k+1}^{*}\right),
\end{align}
where we have made explicit that the algorithm is working with $k+1$ packets by indexing the packet transmission times $t_i$ and $d_i$, with the number of packets as  $t_{i}(k+1)$, and $d_{i}(k+1)$, respectively.

Since $A_{k+1}^{*} = \left(a^*_{1},a^*_{2},a^*_{3},...,a^*_{k+1}\right)$ is such that $a^*_{i}\ge a^*_{i+1}, \ \forall \ i$, we consider the new (sub)-sequence, $A^*_{k} = \{a^*_{2},a^*_{3},...,a^*_{k+1}\}$ of $k$ packets, where 
$T = \sum_{i=2}^{k+1} a^*_i$. Hence with $A^*_{k}$, the $0^{th}$ packet arrives at time $0$, the first packet arrives at time $a_2^*$ and so on, and $a^*_{k+1}$ is the auxiliary time, i.e. after all $k$ packets have arrived till $T$. See Fig. \ref{fig:proof}, for a pictorial description of the construction.
\begin{figure}[h]
	\begin{center}
		\includegraphics[width=3.5in]{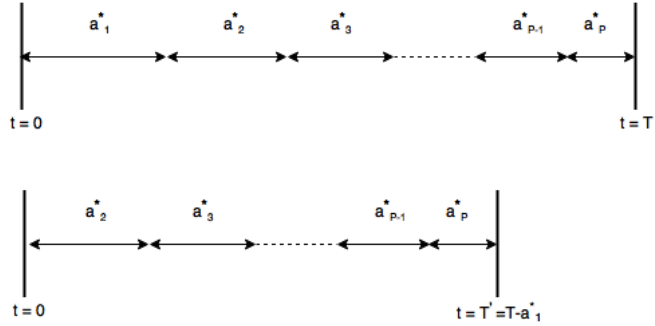}
		\caption{\sl Illustration of the construction used in the induction step in proof of Theorem \ref{thm:main}. \label{fig:proof}} 
	\end{center}
\end{figure}

Since $A_{k}$ has $k$ packets and $a_{i}$'s are decreasing, $A^*_{k}\in \Delta_{D}(k)$. Therefore, from Lemma \ref{lem2} and Lemma \ref{lem:OPTdectime},
\begin{eqnarray}\label{eq:kontime}
\mathsf{ON}\left(A^*_{k}\right)&=&\left(t_{i}(k)\right)_{i=0}^{k-1},\\\label{eq:kOPTtime}
\mathsf{OPT}\left(A^*_{k}\right)&=&\left(d_{i}(k)\right)_{i=0}^{k-1},
\end{eqnarray}
respectively, where more importantly, the corresponding 
\begin{equation}\label{eq:timeeqON}
t_{i+1}(k+1) = t_i(k), i=0,\dots,k-1,
\end{equation} in \eqref{eq:k+1ontime} and \eqref{eq:kontime}, and 
\begin{equation}\label{eq:timeeqOPT}d_{i+1}(k+1) = d_i(k), i=0,\dots,k-1,
\end{equation} in \eqref{eq:k+1OPTtime} and \eqref{eq:kOPTtime}. That is, 
the transmission times to transmit the last $k$ packets of $A_{k+1}^*$ used by both the $\mathsf{ON}$ and 
the $\mathsf{OPT}$ are identical to the transmission times of $k$ packets of  $A_k^*$, respectively. This is the key step of the proof to proceed via induction, that is made possible via Lemma \ref{lem2} and Lemma \ref{lem:OPTdectime} as a result of the worst case input arrival sequence belonging to $\Delta_D$.
Using the induction hypothesis, we have that the competitive ratio of $\mathsf{ON}$ for $A_k^*$ with $k$ packets is bounded by 
 $1+\log\left(k\right)$. 
% Expanding this,
%\begin{eqnarray*}
%\mu_{\text{\text{on}}}\left(A_{k}\right)&\leq &1+\log\left(k\right),\\
%\frac{\sum\limits^{k+1}_{i=2} f\left(t_{i}\right)}{\sum\limits^{k+1}_{i=2} f\left(d_{i}\right)}&\leq &1+\log\left(k\right),\\
%\sum\limits^{k+1}_{i=2} f\left(t_{i}\right)&\leq &\left(1+\log\left(k\right)\right)\sum\limits^{k+1}_{i=2} f\left(d_{i}\right).
%\end{eqnarray*}

Now consider the competitive ratio of $\mathsf{ON}$ for $A_{k+1}^*$ with $k+1$ packets, i.e.,
\begin{eqnarray*}
\mu_{\mathsf{ON}}\left(A^*_{k+1}\right)&=&\frac{\sum\limits _{i=0}^{k}f\left(t_{i}(k+1)\right)}{\sum\limits _{i=0}^{k}f\left(d_{i}(k+1)\right)},
\end{eqnarray*}
which on expanding gives
\begin{eqnarray}\label{eq:dummy1}\mu_{\mathsf{ON}}\left(A^*_{k+1}\right)&=&\frac{f\left(t_{0}(k+1)\right)+\sum\limits _{i=1}^{k}f\left(t_{i}(k+1)\right)}{\sum\limits _{i=0}^{k}f\left(d_{i}(k+1)\right)}.
\end{eqnarray}
Now from \eqref{eq:timeeqON} and \eqref{eq:timeeqOPT}, recall that, $t_{i+1}(k+1) = t_i(k)$ for $i=0,\dots, k-1$, and 
$d_{i+1}(k+1) = d_i(k)$ for $i=0,\dots, k-1$. Hence the corresponding energy functions are also identical, i.e. 
$f(t_{i+1}(k+1)) = f(t_{i}(k))$ and $f(d_{i+1}(k+1) )= f(d_{i}(k))$, for $i=0,\dots, k-1$. Therefore, from \eqref{eq:dummy1}, we have $\mu_{\mathsf{ON}}\left(A^*_{k+1}\right)$
\begin{eqnarray}\nn
&= &\frac{f\left(t_{0}(k+1)\right)+\sum\limits _{i=0}^{k-1}f\left(t_{i}(k)\right)}{f\left(d_{0}(k+1)\right) + \sum\limits _{i=0}^{k-1}f\left(d_{i}(k)\right)},\\ \label{eq:dummy2}
&\leq &\frac{f\left(t_{0}(k+1)\right)+\left(1+\log\left(k\right)\right)\sum\limits _{i=0}^{k-1}f\left(d_{i}(k)\right)}{f\left(d_{0}(k+1)\right) + \sum\limits _{i=0}^{k-1}f\left(d_{i}(k)\right)},
\end{eqnarray}
where the inequality follows from the induction hypothesis that  states that $\mathsf{ON}(A_k^*) \le 1+\log\left(k\right)$. Hence, rewriting \eqref{eq:dummy2},
\begin{eqnarray*}
\mu_{\mathsf{ON}}\left(A^*_{k+1}\right)&=&\left(1+\log\left(k\right)\right) + \\
&&\frac{f\left(t_{0}(k+1)\right)-\left(1+\log\left(k\right)\right)f\left(d_{0}(k+1)\right)}{f\left(d_{0}(k+1)\right) +\sum\limits _{i=0}^{k-1}f\left(d_{i}(k)\right)},\\
&\leq &\left(1+\log\left(k\right)\right)+\frac{f\left(t_{0}(k+1)\right)}{\sum\limits _{i=0}^{k}f\left(d_{i}(k+1)\right)},\\
&\overset{(a)}{\leq}&\left(1+\log\left(k\right)\right)+\frac{f\left(\frac{T}{k+1}\right)}{(k+1)f\left(\frac{T}{k+1}\right)},\\
&=&1+\log\left(k\right)+\frac{1}{k+1},\\
&\overset{(b)}{\leq}&1+\log\left(k+1\right),
\end{eqnarray*}
where in $(a)$ the numerator follows since $t_0(k+1) = \frac{T}{k+1}$ from Lemma \ref{lem2} because of $a_i \ge a_{i+1}$, and the denominator follows from the convexity of $f(.)$, $\frac{1}{k+1}\sum\limits _{i=1}^{k+1}f\left(d_{i}(k+1)\right)\geq f\left(\frac{T}{k+1}\right)$, whereas, (b) follows from the fact that,
$\underset{k}{\overset{k+1}{\int}}\frac{1}{x}dx\geq\frac{1}{k+1}\left(k+1-k\right)=\frac{1}{k+1}$,
$ \log\left(k+1\right)-\log\left(k\right)\geq\frac{1}{k+1}$.
\end{proof}

{\it Discussion:} In this section, we proposed a simple online algorithm that assumes that the future packets arrive at equal time intervals and derived its competitive ratio. Since no information is available about the packet arrival times, it is a natural strategy. We first showed that the worst case input sequence for this algorithm is when the inter-arrival times are decreasing, in which case the algorithm has to idle for some time at the end of each packet transmission. This result was key in deriving the competitive ratio of this algorithm and show that it scales logarithmic in the number of packets, and is independent of the common deadline. 

To the best of our knowledge our theoretical bound on the competitive ratio without assuming anything about the inter-arrival times is the first such result. 
To complete the characterization of online algorithms for this classical packet scheduling problem, a matching lower bound on the competitive ratio would have been useful. However, currently that is beyond the scope of this paper and it is unclear whether $1+\log P$ is the best competitive ratio or not. For similar scheduling and load balancing problems \cite{azar1994line, aspnes1997line, gobel2014online}, the best (theoretically) known competitive ratios also scale logarithmically in the quantity of interest, e.g. number of users/packets, etc.
In the next section, we consider a more general framework, where an additional renewable energy source is available and the objective is to minimize the use of grid energy.

\section{Grid + Energy Harvesting}
In this section, we generalize the packet scheduling problem when there are two sources of energy; conventional (grid) and renewable (EH). The EH energy is stored in a battery, and replenished at each subsequent energy arrival subject to the battery constraints.\footnote{We assume that the battery capacity is large enough and it never overflows.} Once again the object of interest is to minimize the use of grid energy in transmitting  the $P$ packets within common deadline time $T$, but now in the presence of the EH source, thereby exploiting as much EH energy as possible. 

Similar to Remark \ref{rem:unboundedcr}, one can show that if the EH energy arrival epochs and amounts are arbitrary, then the competitive ratio of any online algorithm will be arbitrarily large. For example, if large amount of renewable energy arrives close to the deadline time of $T$, then any online algorithm may not use all of that energy, while the optimal offline algorithm will, making the competitive ratio large.

Thus, we restrict ourselves to the case when the amount of EH energy that arrives at any time $t$ is a random variable that is identically distributed across time, but whose distribution may or may not be known ahead of time to the online algorithm. To exploit the EH energy, 
we propose a natural greedy extension of 
$\mathsf{ON}$, call it $\text{EH}-\mathsf{ON}$, that uses as much EH energy as possible while following the power profile of the earlier proposed online algorithm $\mathsf{ON}$.  As before, the information about energy arrivals and packet arrival times is revealed causally. 

\subsection{Online Algorithm $\text{EH}-\mathsf{ON}$}
The transmission time set by the proposed online algorithm $\text{EH}-\mathsf{ON}$ with EH is identical to the online algorithm $\mathsf{ON}$ without EH. Therefore, the power profile (the power transmitted at any time) of $\text{EH}-\mathsf{ON}$ is identical to that of the $\mathsf{ON}$ algorithm. The only non-trivial decision to make is: which energy source to use at each time to support the power profile set by $\mathsf{ON}$. For that purpose, with $\text{EH}-\mathsf{ON}$, the transmitter follows a greedy policy and uses the renewable energy from the battery for as long as possible to support the power profile of 
$\mathsf{ON}$. The transmitter disconnects from the battery only when there is no energy in the battery and switches over to the grid.

Let $t_{i}$ denote the transmission time of the $i^{th}$ packet from $\mathsf{ON}$, and let $R_{i}$ denote the power (energy/time) used to transmit the $i^{th}$ packet by $\mathsf{ON}$. Let $n(i)$ be the number of renewable energy arrival instants during the transmission of the $i^{th}$ packet, i.e. within time interval $[s_i, f_i]$ set by the $\mathsf{ON}$ algorithm, where the $j^{th}$ EH energy arrival instant happens at time $\tau_{ij}$ with amount $E_{ij},j=1, \dots, n(i)$. Let $E_{i}^{\text{idle}}$ be the total EH energy arrived after the transmission of $i^{th}$ packet and before the start of the transmission of $(i+1)^{th}$ packet. Let $B_{i}$ represent the total energy present in the battery at the start of the transmission of $i^{th}$ packet. Let $B_{\text{max}}$ be the energy capacity of the battery.
Algorithm \text{EH}-$\mathsf{ON}$ describes how the renewable energy is used. The basic idea of this algorithm is to use renewable energy as quickly as possible and for as long as possible, to minimize the grid energy, where $G_i$ represents the grid energy used to transmit packet $i$. The algorithm describes when to use the EH energy and the grid energy, respectively.

\begin{algorithm}

\label{algo3}

	\caption{\text{EH}-$\mathsf{ON}$}
	initialize $t_{0}>>\frac{T}{P}$, $e_{0}=E_{0}^{\text{idle}}=0$, $G_{0}=f\left(t_{0}\right)$\;
	$E_{ik}$ : $k^{th}$ energy arrival during $i^{th}$ packet transmission\;
    $n(i)$ : total number of energy arrivals during $i^{th}$ packet transmission\;
    $E_{i}^{\text{idle}}$ : total harvested energy that arrives after the completion of $i^{th}$ packet but before the start of transmission of the $\left(i+1\right){}^{th}$ packet, its $0$ if there is no idling time\;
	\For(){$i := 1$ to $P$}
   	{
	$t_{i}=\min\left(t_{i-1},\frac{T-\sum\limits^{i-1}_{l=1}a_{l}}{P-i+1}\right)$\;
	Use power $R_{i}=\frac{f\left(t_{i}\right)}{t_{i}} $ to transmit packet $i$\; 
	$B_{i}=\min\left(B_{\text{max}},\, B_{i-1}+E_{i-1}^{\text{idle}}+G_{i-1}-f\left(t_{i-1}\right)\right)$\;
	\For(){$j := 1$ to $n(i)$}
	{
		$B=\max\left(B_{\text{max}},\,\sum\limits _{k=1}^{j-1}E_{ik}+B_{i}\right)$\;
		$e_{j}=\max\left(e_{j-1},\,\tau_{i,j}R_{i}-B\right)$\;
		$w=\frac{\left(\tau_{i,j}-\tau_{i,j-1}\right)R_{i}-\left(e_{j}-e_{j}-1\right)}{R_{i}}$\;
		Use EH source in time interval $\left[\tau_{i,j-1},\,\tau_{i,j-1}+w\right)$\;
		Use Grid energy source in time interval $\left[\tau_{i,j-1}+w,\,\tau_{i,j}\right]$\;
	}
	$B=\max\left(B_{\text{max}},\,\sum\limits _{k=1}^{n(i)}E_{ik}+B_{i}\right)$\;
	$G_{i}=\max\left(e_{n(i)},\, f\left(t_{i}\right)-B \right)$\;
    }
    %$G_{\text{on}}=\sum\limits_{i=1}^{P}G_i$\;
\end{algorithm}
%\begin{algorithm}
%\label{algo3}
%	\caption{\text{OnlineAlgorithm}}
%	initialize $t_{0}>>\frac{T}{P}$, $e_{0}=0$\;
%	\For(){$i := 1$ to $P$}
%   	{
%	$t_{i}=\min\left(t_{i-1},\frac{T-\sum\limits^{i-1}_{l=1}a_{l}}{P-i+1}\right)$, $R_{i}=\frac{f\left(t_{i}\right)}{t_{i}}$, $B\leftarrow \text{Battery energy}$\;
%	\For(){$j := 1$ to $n(i)$}
%	{
%		$e_{j} = \max\left(e_{j-1},\,\tau_{ij}R_{i}-\sum\limits_{k=1}^{j-1}E_{ik}-B\right)$\;
%	}
%	$G_{i}=\max\left(e_{n(i)},\, f\left(t_{i}\right)-\sum\limits_{k=1}^{n(i)}E_{ik}-B\right)$\;
%    }
%    $G_{\text{on}}=\sum\limits_{i=1}^{P}G_i$\;
%\end{algorithm}

%Even though $\text{EH}-\mathsf{ON}$ is a natural extension of $\mathsf{ON}$, however, its not easily amenable for analytical results on its competitive ratio, because of an additional source of uncertainty because of renewable energy. We expect $\text{EH}-\mathsf{ON}$ it to be close to $\mathsf{ON}$ in terms of competitive ratio, however, at the moment we do not know what is its competitive ratio performance.  We numerically compare the performance of the $\text{EH}-\mathsf{ON}$ with an enhanced optimal offline algorithm (that is allowed to use all the renewable energy that arrives over time from the beginning) in the next section, where we conclude that the competitive ratio of $\text{EH}-\mathsf{ON}$ is bounded and is close to $1$ (optimal) similar to $\mathsf{ON}$.

We next show that under some natural assumptions on the EH arrival process, we can show that the $\text{EH}-\mathsf{ON}$ algorithm has a competitive ratio of $c(1+\log P)$, where $c$ is a constant. So essentially, both the $\text{EH}-\mathsf{ON}$ and the $\mathsf{ON}$ have competitive ratios that scale identically in the number of packets $P$. 

Consider the optimal offline algorithm with EH, $\text{EH}-\mathsf{OPT}$. Clearly, for any packet inter-arrival sequence $A$, the total energy (grid + EH) used by 
$\text{EH}-\mathsf{OPT}$ is $E^A_{\text{EH}-\mathsf{OPT}} = E^A_{\mathsf{OPT}}$, where $E^A_{\mathsf{OPT}}$ is the total energy needed by the optimal offline algorithm $\mathsf{OPT}$ in the no EH case. Let the optimal grid energy that $\text{EH}-\mathsf{OPT}$ uses be $G^A_{\text{EH}-\mathsf{OPT}}$.

Let the sum of all the EH energy that arrives in time interval $[0, T/2]$ be $E_{\ell}$, and that arrives in time interval $[T/2, T]$ be $E_{r}$, respectively. Then the following remark is in order.
\begin{remark}\label{rem:nondeg} Let all the EH energy $E_{\ell} + E_{r}$ that actually arrives over several instants in interval $[0,T]$ be made available to $\text{EH}-\mathsf{OPT}$ at time $t=0$ itself. Then  
it follows that 
the grid energy used by the $\text{EH}-\mathsf{OPT}$ is lower bounded by $E^A_{\mathsf{OPT}} - (E_{\ell} + E_{r})$.

To keep the competitive ratio non-trivial, for a fixed packet inter-arrival time sequence $A$, we have to assume that for any realization of EH energy arrivals, 
$$E^A_{\mathsf{OPT}} - (E_{\ell} + E_{r}) > 0.$$ Equivalently this condition implies that only EH energy is not sufficient for the optimal offline algorithm to transmit all the $P$ packets, even if all the EH energy is available at time $0$. If this condition is violated, then any online algorithm cannot be competitive.
\end{remark}

\begin{assumption}\label{asm:1}Let $\eta=\bbE\left(E_{\ell}\right)=\bbE\left(E_{r}\right)$. We assume that 
\begin{equation}\label{eq:asm1}
\eta\leq\frac{(m-1)E^A_{\mathsf{OPT}}}{2m},
\end{equation} for some constant $m>1$ and any $A$. It is a reasonable assumption since the amount of EH energy arriving at any time does not depend on the number of packets $P$, while the total energy needed (grid + EH) $E_{\mathsf{OPT}}$ is increasing in $P$. Since $P$ is typically large,  it is safe to make this assumption. It is also important to make this assumption, since otherwise the actual grid energy used by the optimal offline algorithm $$E^A_{\mathsf{OPT}} - (E_{\ell} + E_{r})$$ can be arbitrarily small, making the online algorithm have arbitrarily bad competitive ratio.
\end{assumption}

\begin{theorem}Under Assumption \ref{asm:1}, the competitive ratio of $\text{EH}-\mathsf{ON}$ is upper bounded by $m(1+\log P)$, for smallest $m>1$ that satisfies Assumption \ref{asm:1}.
\end{theorem}
\begin{proof} From Remark \ref{rem:nondeg}, recall that the total energy $E_{\text{EH}-\mathsf{OPT}}$ used by $\text{EH}-\mathsf{ON}$ is such that $E_{\text{EH}-\mathsf{OPT}} > E_{\ell} + E_{r}$, since otherwise the competitive ratio can be unbounded. 

Let the energy that the algorithm $\mathsf{ON}$ (Section \ref{sec:ON}) uses without any EH energy for packet inter-arrival time sequence $A$ in time interval $[0, T/2]$, be $E^A_{\mathsf{ON},\ell}$ and in time interval $[T/2, T]$ be $E^A_{\mathsf{ON}, r}$, respectively. By the definition of $\mathsf{ON}$, it is easy to follow that for any $A$, $E^A_{\mathsf{ON},\ell} \le E^A_{\mathsf{ON},r}$, since otherwise we can increase the packet transmission times while decreasing the overall energy requirement. 

Recall that the total energy used by $\text{EH}-\mathsf{ON}$ is same as the total energy used by $\mathsf{ON}$ to transmit all the $P$ packets, only $\text{EH}-\mathsf{ON}$ sources some of its energy requirement from the EH source.
Therefore,  from the optimality of $\text{EH}-\mathsf{OPT}$, $$E^A_{\mathsf{ON},\ell} + E^A_{\mathsf{ON},r} \ge 
E_{\text{EH}-\mathsf{OPT}}.$$ 

Hence it follows that 
\begin{eqnarray*}
E^A_{\mathsf{ON},\ell} + E^A_{\mathsf{ON},r} &\ge& E_{\ell} + E_{r},\\
2 E^A_{\mathsf{ON},r} &\stackrel{(a)}\ge& E_{\ell} + E_{r},\\
E^A_{\mathsf{ON},r} &\ge& \frac{E_{\ell} + E_{r}}{2},
\end{eqnarray*}
where $(a)$ follows since $E^A_{\mathsf{ON},\ell} \le E^A_{\mathsf{ON},r}$.
In particular, $$E^A_{\mathsf{ON},r} \ge \frac{E_{\ell}}{2},$$ which implies that the amount of energy used by $\text{EH}-\mathsf{ON}$ in time-interval $[T/2, T]$ is at least half the energy that arrives in interval $[0, T/2]$. Since $\text{EH}-\mathsf{ON}$  is a greedy algorithm in terms of using the EH energy, and all of $E_{\ell}$ is available at time $t=T/2$, it follows that $\text{EH}-\mathsf{ON}$ uses at least $E_{\ell}/2$ amount of EH energy by the deadline $T$. Therefore the grid energy $G_{\mathsf{ON}}$ used by $\text{EH}-\mathsf{ON}$ is at most $E^A_{\mathsf{ON},\ell} + E^A_{\mathsf{ON},r} - E_{\ell}/2$.

Moreover, we know that the grid energy used by $\text{EH}-\mathsf{OPT}$ is at least  $E^A_{\mathsf{OPT}} - (E_{\ell} + E_{r})$ which by definition is positive. Since the 
EH energy arrival process has identical distribution across time, we have that 
\begin{equation}\label{eq:iid}
\bbE\{E_{\ell} + E_{r}\} = 4\bbE\{E_{\ell}/2\}
\end{equation}

 Therefore, the expected competitive ratio for $\text{EH}-\mathsf{ON}$
\begin{eqnarray}\nn
\mu &=&  \frac{\bbE\{G_{\mathsf{ON}}\}}{\bbE\{G_{\text{EH}-\mathsf{OPT}}\}}, \\\label{eq:ehcr}
\mu &\le& \frac{ E^A_{\mathsf{ON},\ell} + E^A_{\mathsf{ON},r} - \bbE\{E_{\ell}/2\} }{E^A_{\mathsf{OPT}} - \bbE\{E_{\ell} + E_{r}\}}.
\end{eqnarray}

From Assumption \ref{asm:1},
\begin{equation}\label{eq:asm1}
\eta\leq\frac{(m-1)E^A_{\mathsf{OPT}}}{2m}
\end{equation} for some constant $m$ and any $A$.  Therefore, multiplying and dividing by $(1+\log P)$ and subtracting $1/2$ from the denominator in \eqref{eq:asm1}, we get
\begin{eqnarray}
\eta\leq\frac{(m-1)(1+\log P)E^A_{\mathsf{OPT}}}{2m(1+\log P)-\frac{1}{2}},
\end{eqnarray}
which on rewriting is equivalent to 
\begin{equation}\label{eq:dummy}
\frac{\left(1+\log P\right)E^A_{\mathsf{OPT}}-\frac{\eta}{2}}{E^A_{\mathsf{OPT}}-2\eta}\leq m\left(1+\log P\right).
\end{equation}
From Theorem \ref{thm:main}, we know that $\frac{E^A_{\mathsf{ON},\ell} + E^A_{\mathsf{ON},r}}{E^A_{\mathsf{OPT}}} \le 1+\log P$, and hence the RHS of \eqref{eq:dummy} is equal to LHS of \eqref{eq:ehcr}, and we get that for some constant $m > 1$, 
\begin{eqnarray}\label{eq:ehcrfinal}
\mu \le m\left(1+\log P\right).
\end{eqnarray}
{\it Discussion:} In this section, we considered the case when energy from both the grid and an EH source is available. In this scenario, the biggest challenge for any online algorithm is to ensure that enough EH energy is used up and the leftover EH energy is minimized, since the optimal offline algorithm is going to completely use up all the EH energy.  To keep the competitive ratio non-trivial, we assumed that the total energy arriving from the EH source is not too large and the optimal offline algorithm has to use 'significant' grid energy to transmit all the $P$ packets. 

To ensure that enough EH energy is used, we proposed a greedy extension of the ${\mathsf{ON}}$ algorithm that uses the same power transmission profile and transmission times as prescribed by the ${\mathsf{ON}}$ algorithm, and sources its energy requirement from the EH source as long as possible, otherwise uses the grid energy source.
Under this assumption, we showed that the new online algorithm $\text{EH}-{\mathsf{ON}}$ at least uses half of the EH energy arriving in the first half of the deadline time, while the optimal algorithm can at most use all the energy that arrives till the deadline time. Since the EH energy arrivals are identically distributed across time, this allows us to reuse the competitive ratio bound that we derived on the ${\mathsf{ON}}$ algorithm to show that the competitive ratio of 
$\text{EH}-{\mathsf{ON}}$ is at most some constant times $(1+\log P)$. Thus, both the ${\mathsf{ON}}$ and its greedy extension have the same scaling in the competitive ratio as a function of the number of packets. 

\end{proof}

\section{Simulations} In this section, we provide numerical results to better understand the competitive ratio of 
$\mathsf{ON}$. For all simulations without energy harvesting, we assume that the packet inter-arrival times 
$a_i$ are exponentially distributed with mean $T/P$, and $T=100$ secs, packet size $B=200$kb and the number of packets are taken to be $200$. Moreover, noise power spectral density is taken to be $10^{-19}$ Watt/Hz and bandwidth $=1$MHz. In Fig. \ref{fig:noEHcompratio}, we plot the (simulated) competitive ratio of $\mathsf{ON}$ together with the theoretical bound of $1+\log P$. We see that the competitive ratio of $\mathsf{ON}$ is close to $1$ (optimal), and much smaller than the  theoretical bound of $\log P$. 

In Fig. \ref{fig:worstcaseP}, we also plot the competitive ratio for the worst case sequence $a_i \ge a_{i+1}$ as a function of $P$, where for each value of $P$, the worst values of $a_i$ are found via optimization. In particular, we start with inter-arrival times 
$a_i$ are exponentially distributed with mean $T/P$, and then steer the inter-arrival times in the direction of increasing the competitive ratio via gradient descent algorithms.
 We restrict to small values of $P$, since otherwise the optimization for finding the worst case $a_i$ is prohibitive. Even for this case, the competitive ratio of $\mathsf{ON}$ is fairly close to $1$.
 
In Figs. \ref{fig:noEHdeadline} and \ref{fig:noEHpacketsizeB}, we plot the competitive ratio performance of $\mathsf{ON}$ while varying the deadline times $T$ and packet sizes $B$, together with the theoretical upper bound. From all the figures it is clear that $\mathsf{ON}$ performs very close to the optimal.

For the hybrid energy arrival scenario, we assume that the packet inter-arrival times 
$a_i$ are exponentially distributed with mean $T/P$, EH energy inter-arrival epochs are exponentially distributed with mean $T/(N+1)$, where $N=20$ (if not varied) is the total number of EH epochs. Moreover, the amount of energy arrival at each EH epoch is also exponentially distributed with mean $10$ mJ. Once again we use $T=100$ secs, while larger packet size of $B=500$ kb and the number of packets are taken to be $400$. In Fig. \ref{fig:EHP}, we plot the (simulated) competitive ratio of 
$\mathsf{\text{EH}-ON}$ together and observe that similar $\mathsf{ON}$ it is very close to the optimal. In Fig. 
\ref{fig:EHavgenergy} we plot the competitive ratio of 
$\mathsf{\text{EH}-ON}$ as a function of average energy harvested.

\begin{figure}[h]
	\begin{center}
		\includegraphics[height=2.75in]{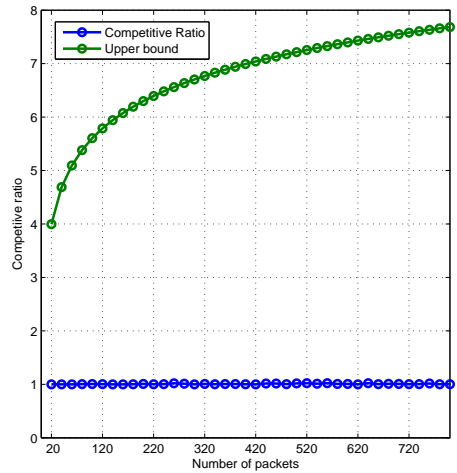}
		\caption{\sl Competitive ratio of the $\mathsf{ON}$ with different number of packets $(P)$ and the theoretical upper bound. \label{fig:noEHcompratio}} 
	\end{center}
\end{figure}

\begin{figure}[h]
	\begin{center}
		\includegraphics[height=2.75in]{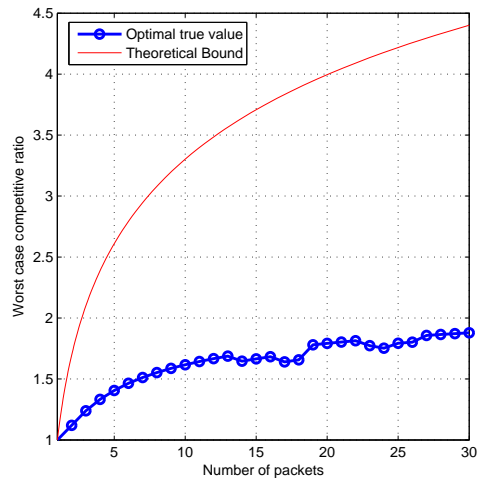}
		\caption{\sl Competitive ratio of the $\mathsf{ON}$ with different number of packets $(P)$ under worst case input of inter-arrival time $a_i > a_{i+1}$. \label{fig:worstcaseP}} 
	\end{center}
\end{figure}

\begin{figure}[h]
	\begin{center}
		\includegraphics[height=2.75in]{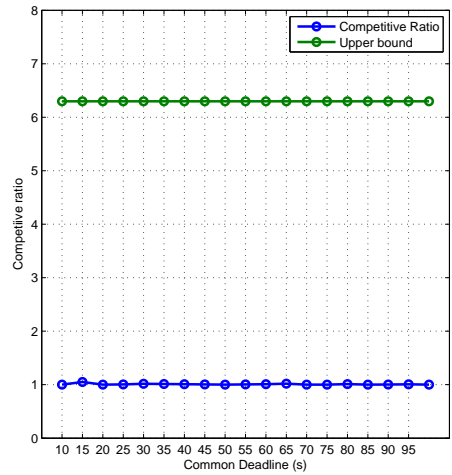}
		\caption{\sl Competitive ratio of the $\mathsf{ON}$ with different deadline times $(T)$ and the theoretical upper bound. \label{fig:noEHdeadline}} 
	\end{center}
\end{figure}

\begin{figure}[h]
	\begin{center}
		\includegraphics[height=2.75in]{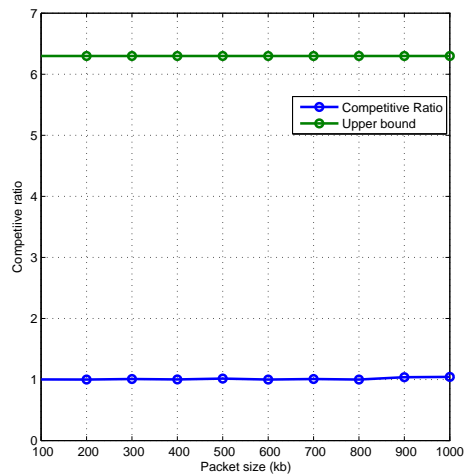}
		\caption{\sl Competitive ratio of the $\mathsf{ON}$ with different packet sizes $(B)$ and the theoretical upper bound.. \label{fig:noEHpacketsizeB}} 
	\end{center}
\end{figure}

\begin{figure}[h]
	\begin{center}
		\includegraphics[height=2.75in]{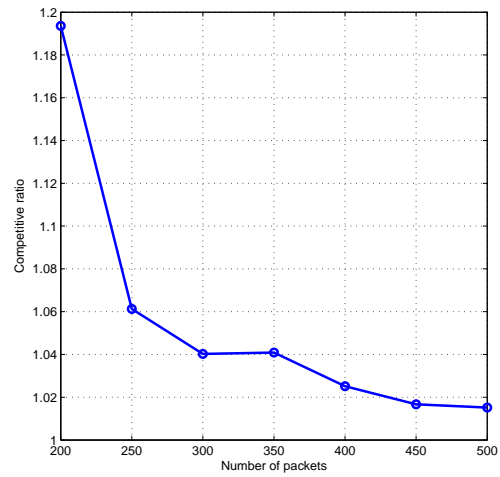}
		\caption{\sl Competitive ratio of the $\mathsf{\text{EH}-ON}$ with different number of packets $P$. \label{fig:EHP}} 
	\end{center}
\end{figure}

\begin{figure}[h]
	\begin{center}
		\includegraphics[height=2.75in]{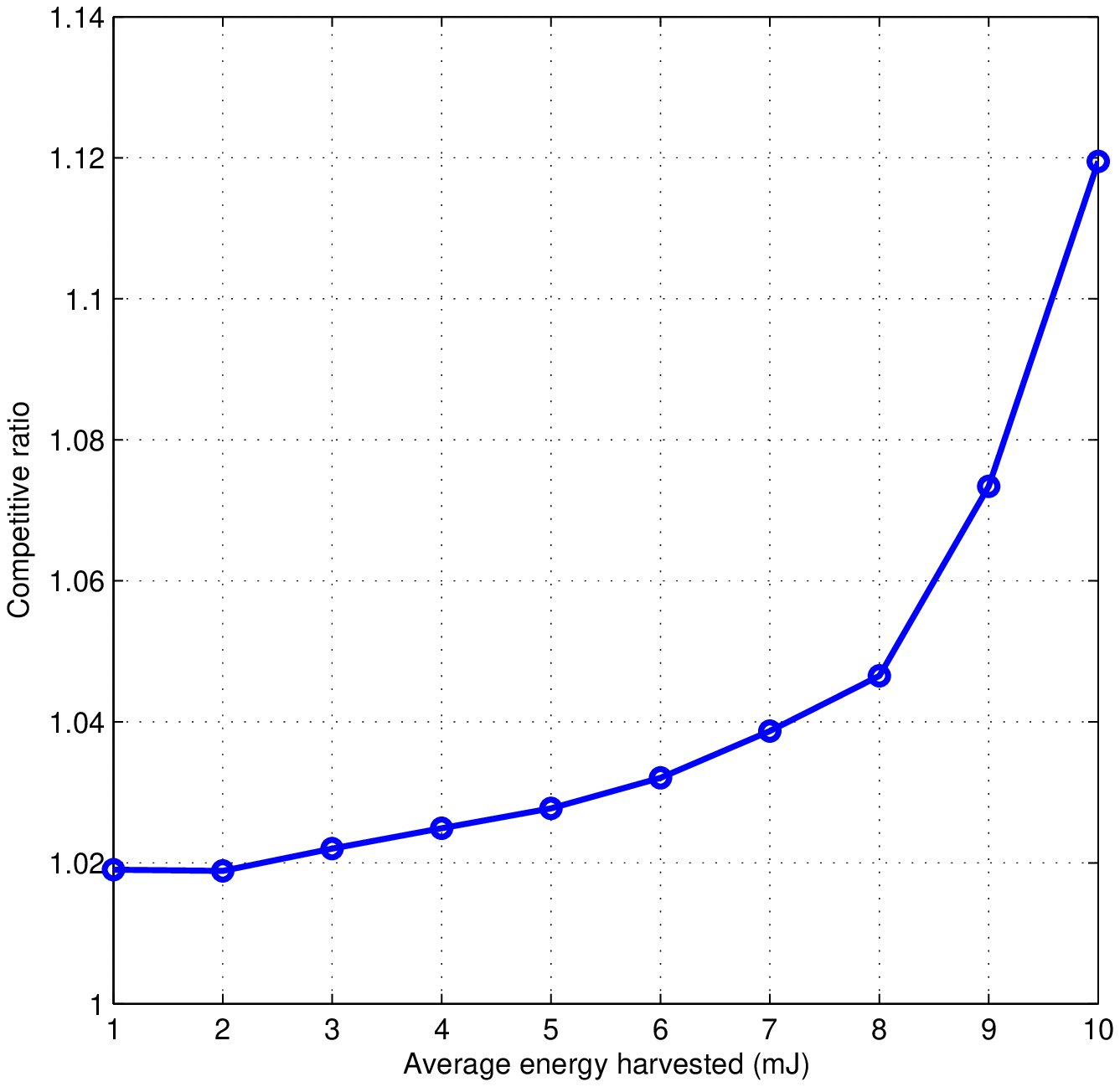}
		\caption{\sl Competitive ratio of the $\mathsf{\text{EH}-ON}$ with different average energy harvested. \label{fig:EHavgenergy}} 
	\end{center}
\end{figure}

\section{Conclusions} In this paper, we considered the online setting of a classical problem of minimizing energy for transmitting multiple packets given a common deadline, without making any assumptions on the packet inter-arrival times. We showed that even for this most general input model, the proposed algorithm $\mathsf{ON}$, has a competitive ratio that only grows logarithmically with the number of packets and is independent of the common deadline. The simulated performance of the proposed algorithm is far better than the theoretically guaranteed performance, and for most cases it is very close to the optimal. Thus, a natural question that remains open is : whether the competitive ratio analysis of the $\mathsf{ON}$ can be tightened to show that it is a constant, or can a lower bound be derived that shows that no online algorithm can have competitive ratio smaller than logarithm of the number of packets. For the hybrid energy case, where both conventional and renewable energies are available, we show that a natural greedy extension of 
$\mathsf{ON}$ has very similar theoretical performance, is and very close to the optimal numerically.

\bibliographystyle{IEEEtran}
\bibliography{bibfile}

\end{document}